\pdfoutput=1
\RequirePackage{ifpdf}
\ifpdf 
\documentclass[pdftex]{sigma}
\else
\documentclass{sigma}
\fi

\usepackage{stmaryrd}

\newcommand{\R}{\mathbb{R}}
\newcommand{\E}{\mathbb{E}}
\newcommand{\I}{\mathbb{I}}
\newcommand{\CR}{\mathcal{R}}
\newcommand{\diag}{\operatorname{diag}}

\begin{document}
\allowdisplaybreaks

\newcommand{\arXivNumber}{2006.08301}

\renewcommand{\PaperNumber}{083}

\FirstPageHeading

\ShortArticleName{On Products of Delta Distributions and Resultants}

\ArticleName{On Products of Delta Distributions and Resultants}

\Author{Michel BAUER~$^{\dag^1\dag^2\dag^3\dag^4\dag^5}$ and Jean-Bernard ZUBER~$^{\dag^6\dag^7}$}

\AuthorNameForHeading{M.~Bauer and J.-B.~Zuber}

\Address{$^{\dag^1}$~Institut de Physique Th\'eorique de Saclay, CEA-Saclay, F-91191 Gif-sur-Yvette, France}
\EmailDD{\href{mailto:michel.bauer@ipht.fr}{michel.bauer@ipht.fr}}
\Address{$^{\dag^2}$~CNRS, UMR 3681, IPhT, F-91191 Gif-sur-Yvette, France}
\Address{$^{\dag^3}$~D\'epartement de math\'ematiques et applications, \'Ecole normale sup\'erieure,\\
\hphantom{$^{\dag^3}$}~F-75005 Paris, France}
\Address{$^{\dag^4}$~CNRS, UMR 8553, DMA, ENS, F-75005 Paris, France}
\Address{$^{\dag^5}$~PSL Research University, F-75005 Paris, France}

\Address{$^{\dag^6}$~Sorbonne Universit\'e, UMR 7589, LPTHE, F-75005, Paris, France}
\EmailDD{\href{mailto:jean-bernard.zuber@upmc.fr}{jean-bernard.zuber@upmc.fr}}
\Address{$^{\dag^7}$~CNRS, UMR 7589, LPTHE, F-75005, Paris, France}

\ArticleDates{Received June 16, 2020, in final form August 20, 2020; Published online August 25, 2020}

\Abstract{We prove an identity in integral geometry, showing that if $P_x$ and $Q_x$ are two polynomials, $\int {\rm d}x\, \delta(P_x) \otimes \delta(Q_x)$ is proportional to~$\delta(R)$ where $R$ is the resultant of~$P_x$ and~$Q_x$.}


\Classification{46F10; 49Q15}

\section{Introduction}

In their classical text \cite[see in particular Chapter~3]{GS}, Gelfand and Shilov introduce generalized functions localized on smooth submanifolds of $\mathbb{R}^n$. A most important example is $\delta(f)$ with $f\colon \mathbb{R}^n\to \mathbb{R}$ a $C^1$ function whose gradient $\operatorname{Grad} f$ vanishes nowhere on $f=0$.

In this note, we introduce a variant of the original definition of $\delta(f)$ and extend it to deal with certain cases when $f$ and $\operatorname{Grad} f$ have common zeros. A salient feature of this extension is that $\delta(f)$, when defined, is no longer always a distribution, but a positive $\sigma$-finite measure.

Our aim is to prove the following combinatorial identity:
\begin{proposition} \label{prop:main}
 Let $A$, $B$ be disjoint finite sets. For $x\in \mathbb{R}$, $a,b\in \mathbb{R}^*$, $u:=(u_\alpha)_{\alpha \in A} \in \mathbb{R}^A$ and $v:=(v_\beta)_{\beta \in B} \in \mathbb{R}^B$, set $P_x(u):=a\prod_{\alpha \in A}(x-u_\alpha)$ and $Q_x(v):=b\prod_{\alpha \in B}(x-v_\beta)$. Let $R(u,v):=a^{|B|} b^{|A|} \prod_{\alpha \in A \atop \beta\in B} (u_\alpha-v_\beta)$ be the resultant of $P_x$ and $Q_x$ seen as polynomials in the variable $x$. Set
\[ J(u,v):= b^{|A|-1} a^{|B|-1} \sum_{\beta' \in B} \prod_{\beta'' \neq \beta'}
 \frac{\prod\limits_{\alpha''\in A}(v_{\beta''}-u_{\alpha''})}{(v_{\beta'}-v_{\beta''})},\]
and
\[ \tilde{J}(u,v):= a^{|B|-1} b^{|A|-1} \sum_{\alpha' \in A} \prod_{\alpha'' \neq \alpha'} \frac{\prod\limits_{\beta''\in B}(u_{\alpha''}-v_{\beta''})}{(u_{\alpha'}-u_{\alpha''})}, \]
which is nothing but $J$ with the roles of $P_x$ and $Q_x$ interchanged.
Then:
\begin{enumerate}\itemsep=0pt
\item[$(i)$] $J$ and $\tilde{J}$ differ (at most) by a sign, precisely $J(u,v):=(-)^{|A||B|-1} \tilde{J}(u,v)$.
\item[$(ii)$] $J$ is a polynomial in $(u,v)\in \mathbb{R}^{A \cup B}$.
\item[$(iii)$] The $\delta$-identity
\begin{gather} \label{eq:deltaident} \int_{\mathbb{R}} {\rm d}x\, \delta(P_x(u)) \otimes \delta(Q_x(v))= |J(u,v)| \delta(R(u,v)),
\end{gather}
$($involving $J$ as the multiplier$)$ holds, where for each $x\in \mathbb{R}$, $\delta(P_x(u))$ and $\delta(Q_x(v))$ are positive $\sigma$-finite measures on $\mathbb{R}^A$ and $\mathbb{R} ^B$ respectively and $\delta(R(u,v))$ is a positive $\sigma$-finite measure on $\mathbb{R}^{A\cup B}=\mathbb{R}^A \times\mathbb{R} ^B$.
\end{enumerate}
\end{proposition}

The main point is the combinatorial identity equation~\eqref{eq:deltaident} to which we refer as the $\delta$-identity in the sequel. Before we can prove Proposition~\ref{prop:main}, we need to make sense of the $\delta$ functions that appear on each side of equation~\eqref{eq:deltaident}. This is the goal of the next section where they are defined as positive $\sigma$-finite measures. First, we make a remark on our motivations.

\begin{remark}
The $\delta$-identity is grounded in integral geometry but it proves useful in the following context occurring in probability theory or in statistical mechanics and which we describe informally using a physicist's notation.
Consider a random vector $X=(X_1,X_2,\dots, X_{d})$, taking its values in a compact convex set $K \subset \R^{d}$, with a density $\mu(x_2,\dots,x_{d})$ on $K$ which is independent of~$x_1$;
consider two functions~$P$ and~$Q$ on~$K$ that are polynomial in~$x_1$.
The joint probability density function (PDF) of $P(X)$ and $Q(X)$ may be written as
\begin{gather} \label{deldel}\rho(p,q) = \E( \delta(P(X)-p) \delta(Q(X)-q)) . \end{gather}
We rewrite equation~\eqref{deldel} as
\begin{gather}\label{deldel2} \rho(p,q) =\int {\rm d}x_2\cdots {\rm d}x_{d}\, \mu(x_2,\dots,x_{d})\int {\rm d}x_1\, \delta(P(x)-p) \delta(Q(x)-q)\mathbf{1}_{x\in K} . \end{gather}
Suppose that for given $(x_2,\dots,x_{d})\in \R^{d-1}$ either the interval $\{x
_1\in \R, (x_1,x_2,\dots,x_{d})\in K\}$ is empty, or it contains all the roots of the two polynomials $P(x)-p$ and $Q(x)-q$ of the variable $x_1$. Under this assumption, first one can dispense with the explicit mention of $\mathbf{1}_{x\in K}$ in equation~\eqref{deldel2} and then second one can use the identity equation~\eqref{eq:deltaident} to write the resulting $x_1$ integral over $\R$ in terms of $\delta(R)$, where $R$ is the resultant of $P(x)-p$ and $Q(x)-q$ in the variable~$x_1$, and of~$J$, the expression defined in Proposition~\ref{prop:main}.

This situation is encountered in the following particular case of
Horn's problem. Consider two $3\times 3$ real traceless symmetric matrices~$A$ and~$B$, acted upon by conjugation by matrices of SO(3).
We take $A$ and $B$ as random variables, independently and uniformly distributed on their orbits characterised by their eigenvalues $\alpha=(\alpha_1,\alpha_2,\alpha_3)$ and $\beta=(\beta_1,\beta_2,\beta_3)$. Then the characteristic polynomial of their sum $C=A+B$ is
\[\det(z\I_3-C)=\det\big(z\I_3 - \diag(\alpha) -\CR\diag(\beta)\CR^T\big)= z^3 + z P(\CR) + Q(\CR) , \]
where $\CR$ $\in$ SO(3) is taken randomly and uniformly according to the normalised Haar measure. Following equation~\eqref{deldel}, we write the PDF of $P$ and $Q$ in the form
\[ \rho(p,q) =\int D\CR \delta(P(\CR)-p) \delta(Q(\CR)-q) . \]
For example, if $\CR$ is expressed in terms of Euler angles $(\theta,\phi,\psi)$, a simple calculation shows that~$P(\CR)$ and~$Q(\CR)$ are (degree 2) polynomials in $c:=\cos\theta$, and
\begin{gather*}\rho(p,q) = \frac{1}{2\pi^2} \int_0^\pi {\rm d}\phi \int_0^\pi {\rm d}\psi \int_{-1}^1 {\rm d}c\, \delta(P(c,\phi,\psi)-p) \delta(Q(c,\phi,\psi)-q) . \end{gather*}
This is precisely where identity equation~\eqref{eq:deltaident} simplifies things a great deal, reducing the previous integral to a~$(\phi,\psi)$-integral of $\delta(R(\phi,\psi)) J$, where~$R$ is the resultant of the two polynomials $P-p$ and $Q-q$, and~$J$ is the Jacobian given above in Proposition~\ref{prop:main}. One of the two integrations over~$\phi$,~$\psi$ localizes on the zeros of $R$ belonging to the appropriate interval, and one is left with a single integral over the remaining variable. This method has been used in~\cite{CZ2} to analyse the divergences that (somewhat unexpectedly) arise in the PDF~$\rho(p,q) $. We refer the reader to that article for further details. That calculation, carried out in the simple case of two degree~2 polynomials, was the source of inspiration for the more general formula equation~\eqref{eq:deltaident}.
\end{remark}

\section{Preliminaries}

Mainly to fix notation, we recall the definition of the $\delta$ measure.

\begin{definition}
 The measure $\delta_x$, the unit mass concentrated at point $x \in \mathbb{R}^n$, is $\delta_x(M):=\mathbf{1}_{x\in M}$ for $M$ an arbitrary Borel subset of $\mathbb{R}^n$.
\end{definition}

\begin{remark} \label{rem:fdel}
An immediate consequence is that if $x \in \mathbb{R}^n$ and $\varphi\colon \mathbb{R}^n \to \mathbb{R}$ is any Borel function, then $\int_{\mathbb{R}^n} \delta_x \varphi =\varphi(x)$. Henceforth, if $f\colon \mathbb{R}^n \to \mathbb{R}$ is any Borel function, $f \delta_x= f(x) \delta_x$ where the left-hand side is the product of the measure $\delta_x$ by the function $f$ and the right-hand side the product of the measure~$\delta_x$ by the constant~$f(x)$.
\end{remark}

The measure $\delta_x$ is a positive $\sigma$-finite measure. Moreover, the measure of any bounded set is finite. Thus $\delta_x$ also defines a distribution: if $\varphi\colon \mathbb{R}^n\to \mathbb{R}$ is $C^\infty$ with compact support, then $\langle \delta_x,\varphi \rangle:= \varphi(x)$, where the bracket can be interpreted as the integral of $\varphi$ against the measure~$\delta_x$. This is of course the definition of the $\delta$ distribution found in any textbook.

The support of $\delta_x$ is a single point. Our aim is to define positive $\sigma$-finite measures whose support is a finite union of affine hyperplanes in~$\mathbb{R}^n$ and which are natural extensions of $\delta_x$. We do this in three steps. Our definition parallels closely the construction in~\cite{GS}. We shall briefly comment on the similarities and differences in Remark~\ref{rem:simdif}.

-- First step: we define $\delta(f)$ when $f\colon \mathbb{R} \to \mathbb{R}$ is a regular function. The one dimensional case is mainly for motivation and also facilitates later comparison with the construction in \cite{GS}.

\begin{definition} \label{def:fregdimone} A function $f\colon \mathbb{R} \to \mathbb{R}$ is called regular if $f$ is $C^1$ and~$f$,~$f'$ have no common zero.
\end{definition}

\begin{definition} \label{def:dimone} Let $f\colon \mathbb{R} \to \mathbb{R}$ be a regular function. The measure~$\delta(f)$ is defined by
\[ \delta(f):=\sum_{x\in \mathbb{R},\, f(x)=0} \frac{1}{|f'(x)|} \delta_x .\]
\end{definition}

If $f\colon \mathbb{R} \to \mathbb{R}$ is regular, the number of its zeros in any bounded set is finite (else the set of zeros would accumulate at a finite point, where necessarily $f=f'=0$) and the derivative is nonzero there. Thus $\delta(f)$ is a positive $\sigma$-finite measure such that the measure of any bounded Borel set is finite. Hence~$\delta(f)$ can also be seen as a distribution: if $\varphi\colon \mathbb{R}\to \mathbb{R}$ is $C^\infty$ with compact support, then~$\langle \delta(f),\varphi \rangle:= \sum_{x\in \mathbb{R},\, f(x)=0} \frac{1}{|f'(x)|} \varphi(x)$, a finite sum in fact (but the number of terms may depend on the support of $\varphi$).

\begin{remark} \label{rem:conv}
If $\theta_k\colon \mathbb{R} \to \mathbb{R}$, $k\in \mathbb{N}$ is a sequence of Borel functions such that the sequence of measures $\theta_k(t) \,{\rm d}t$ converges weakly to $\delta_0$ (that is, if $\lim\limits_{k\to +\infty} \int_{\mathbb{R}} \varphi(t)\theta_k(t) \, {\rm d}t = \varphi(0)$ for every continuous function $\varphi\colon \mathbb{R}\to \mathbb{R}$ with compact support), and if $f_{a,x}\colon \mathbb{R} \to \mathbb{R}$, $t\mapsto a (t-x)$ where $x\in \mathbb{R}$ and $a\in \mathbb{R}^*$ (the condition under which $f_{a,x}$ is regular), then $\delta(f_{a,x})=\frac{1}{|a|}\delta_x$ is the weak limit of the sequence of measures $\theta_k(f_{a,x}(t))\,{\rm d}t=\theta_k\circ f_{a,x} (t) \, {\rm d}t$. This explains the presence of the denominators in the definition of $\delta(f)$ and motivates that $\delta\circ f$ would also be an appropriate notation.
\end{remark}

The following simple result plays a crucial role:

\begin{lemma} If $f,g\colon \mathbb{R} \to \mathbb{R}$ are regular functions and the loci $f=0$ and $g=0$ do not intersect, then $h=fg$ is regular and
\begin{gather} \label{eq:deltaprod} \delta(fg)=\frac{1}{|f|}\delta(g) +\frac{1}{|g|}\delta(f).\end{gather}
More generally if $f_1,\dots,f_k$ are regular functions with no pairwise common zeros then the product $f_1\cdots f_k$ is regular and
\begin{gather} \label{eq:deltamultiprod} \delta(f_1\cdots f_k)=\sum_{j=1}^k \frac{1}{\Big|\prod\limits_{i\neq j} f_i\Big|}\delta(f_j). \end{gather}
\end{lemma}

\begin{proof} For the first formula, the term $\frac{1}{|f|}\delta(g)$ on the right-hand side for instance is well defined as the product of a measure $\delta(g)$ by a measurable function $\frac{1}{|f|}$ which is finite in a neighborhood of the support of~$\delta(g)$. Then the validity of the first formula is a consequence of $(fg)'(x) =f(x)g'(x)$ when~$x$ is a zero of~$g$, and of Remark~\ref{rem:fdel}. The second formula can be justified in a similar way or from the first by recursion.
\end{proof}

-- Second step: we define $\delta(f)$ for non-constant affine functions in arbitrary dimensions.

\begin{definition}\label{def:aff} If $f\colon \mathbb{R}^n\to \mathbb{R}$ is a non-constant affine function then there is a non-constant linear form $l_1\neq 0$ on $\mathbb{R}^n$ and a constant $x\in\mathbb{R}$ such that $f=l_1-x$. One can find $n-1$ other linear forms $l_2,\dots,l_n$ such that $(l_1,\dots,l_n)$ is a basis of linear forms on~$\mathbb{R}^n$. Consequently $z\colon \mathbb{R}^n \to \mathbb{R}^n$, $y:=(y_1,\dots,y_n) \mapsto z(y)$ where $z(y):=(l_1(y),\dots,l_n(y))$ defines a non-singular linear change of coordinates. If~$M$ is a~Borel subset of $\mathbb{R}^n$, we define $\widetilde{M}\subset\mathbb{R}^{n-1}$ by $z(M) \cap \{x\}\times\mathbb{R}^{n-1}=: \{x\}\times\widetilde{M} $.

We define the measure $\delta(f)$ by
\[ \delta(f)(M):= \int_{\widetilde{M}} |\operatorname{Jac}(y,z)| \, {\rm d}z_2\cdots {\rm d}z_n,\]
where $\operatorname{Jac}(y,z):=\det \big(\frac{\partial y_i}{\partial z_j}\big)_{ij}=1/\operatorname{Jac}(z,y)$ is the standard Jacobian, a constant that could be pulled out of the integral.
\end{definition}

The fact that the definition of $\delta(f)$ is intrinsic, i.e., that $\delta(f)(M)$ is independent on how~$l_1$ is completed into a basis of linear forms, is guaranteed by the multiplicativity of the Jacobian determinant. We shall compare this definition of $\delta(f)$ with its analog defined by Gelfand and Shilov in the context of generalized functions in Remark~\ref{rem:simdif}. For the moment, let us simply observe that~$\delta(f)$ is invariant under translations along the hyperplane $\{l_1=0\}$ and its support is $\{f=0\}$. Moreover, if $n=1$ we may write $f=f_{a,x/a}$ as in Remark~\ref{rem:conv} and a simple computation shows that the definitions of $\delta(f)$ in Definitions~\ref{def:dimone} and~\ref{def:aff} coincide.

An immediate consequence of the definition is that if $\varphi\colon \mathbb{R}\to \mathbb{R}$ is a positive Borel function then
 \[ \int_{\mathbb{R}^n} \delta(f) \varphi = \int_{\mathbb{R}^{n-1}} \tilde{\varphi}(x,z_2,\dots,z_n) |\operatorname{Jac}(y,z)| \, {\rm d}z_2\cdots {\rm d}z_n, \]
where $\tilde{\varphi}:=\varphi \circ z^{-1}$.

If $M$ is a bounded Borel subset of $\mathbb{R}^n$, then $\widetilde{M}$ is also bounded, so the positive measure $\delta(f)$ is such that $\delta(f)(M)$ is finite. Hence $\delta(f)$ can also be seen as a distribution: if $\varphi\colon \mathbb{R}\to \mathbb{R}$ is~$C^\infty$ with compact support, then{\samepage
\begin{gather} \label{eq:aff} \langle \delta(f), \varphi \rangle := \int_{\mathbb{R}^{n-1}} \tilde{\varphi}(x,z_2,\dots,z_n) |\operatorname{Jac}(y,z)| \, {\rm d}z_2\cdots {\rm d}z_n, \end{gather}
where $\tilde{\varphi}:=\varphi \circ z^{-1}$ is again $C^\infty$ with compact support.}

If $f\colon \mathbb{R}^n\to \mathbb{R}$ is a non-constant affine function and $a\in \mathbb{R}^*$ then~$af$ is a non-constant affine function and $\delta(af)=\frac{1}{|a|} \delta(f)$.

-- Third and final step: we define~$\delta(f)$ for certain products of non-constant affine functions.

\begin{definition} Let $\mathcal{H}(\mathbb{R}^n)$ be the space of functions $f\colon \mathbb{R}^n\to \mathbb{R}$ that can be written as a product $f=f_1\cdots f_k$ of non-constant, pairwise non-proportional, affine functions $f_1,\dots,f_k$ from~$\mathbb{R}^n$ to~$\mathbb{R}$.

We say that $(f_1,\dots,f_k)$ is a factorization of $f\in \mathcal{H}(\mathbb{R}^n)$.
\end{definition}

We note that generically $f_i$ and $f_j$, $i\neq j$ have common zeros (an affine subspace of dimension $n-2$ unless the $0$ loci are parallel hyperplanes) so that $\{f=0\}$ is not a smooth submanifold of~$\mathbb{R}^n$.

\begin{definition} \label{def:delta-prod} Suppose that $f\in \mathcal{H}(\mathbb{R}^n)$ and let $(f_1,\dots,f_k)$ be a factorization of~$f$. We define
\[ \delta(f):=\sum_{j=1}^k \frac{1}{\Big|\prod\limits_{i\neq j} f_i\Big|}\delta(f_j)\]
with $\delta(f_j)$ as in Definition~\ref{def:aff}.
\end{definition}

It is readily seen that the right-hand side in the definition is still a positive measure on $\mathbb{R}^n$, and it is $\sigma$-finite because $\Big\{\prod_{i\neq j} f_i=0\Big\}$ is a finite union of affine subspaces of dimension $n-2$ in the hyperplane $f_j=0$ which is the support of $\delta(f_j)$.

We have to check that $\delta(f)$ is defined intrinsically:

\begin{lemma}For $f\in \mathcal{H}(\mathbb{R}^n)$, the right-hand side in the definition of $\delta(f)$ does not depend on the factorization of~$f$.
\end{lemma}

\begin{proof} If $(g_1,\dots,g_l)$ is another factorization of $f\in \mathcal{H}(\mathbb{R}^n)$ then $l=k$, there is a unique permutation $\sigma$ of $\llbracket 1,k\rrbracket$ and unique $\lambda_1,\dots,\lambda_k \in \mathbb{R}^*$ such that $g_{\sigma(j)}=\lambda_j f_j$ for $j=1,\dots,n$. Moreover $\lambda_1\cdots\lambda_k=1$. We observe that for each $j\in \llbracket 1,k\rrbracket$
\[ \frac{1}{\Big|\prod\limits_{i\neq \sigma(j)} g_i\Big|}\delta(g_{\sigma(j)})=\frac{1}{\Big|\prod\limits_{i\neq j} f_i\Big|}\delta(f_{j})\]
because $\delta(\lambda_jf_{j})=\frac{1}{|\lambda_j|}\delta(f_{j})$ and $\lambda_1\cdots\lambda_k=1$, establishing the claim that $\delta(f)$ is defined intrinsically.
\end{proof}

The space $\mathcal{H}(\mathbb{R})$ is simply the space of non-constant real polynomials whose roots are real and simple. Thus in dimension~$1$, the condition that the affine functions in the factorization be non-proportional is equivalent to the condition that $f$ be regular. It is immediate to check that for $f\in \mathcal{H}(\mathbb{R})$, the definitions of $\delta(f)$ in Definitions~\ref{def:dimone} and~\ref{def:delta-prod} coincide.

If $f$, $g$ and $fg$ belong to $\mathcal{H}(\mathbb{R}^n)$, it is immediate to check that
\begin{gather*} \delta(fg)=\frac{1}{|f|}\delta(g) +\frac{1}{|g|}\delta(f).\end{gather*}
which is formally identical to equation~\eqref{eq:deltaprod}. The extension to a product of a finite number of factors leads to a counterpart of equation~\eqref{eq:deltamultiprod}.

In fact, once $\delta(f)$ is defined when $f$ is a non-constant affine function, the definition we have given of $\delta(f)$ for $f\in \mathcal{H}(\mathbb{R}^n)$ could be seen as a formal extension of the identity equation~\eqref{eq:deltamultiprod} in this more general context. However, there is one important difference between the two situations:

If $f\in \mathcal{H}(\mathbb{R}^n)$, the $\delta(f)$-measure of a compact set is clearly finite if it does not contain common zeros of $f_i$ and $f_j$ for any $i\neq j$, but may be infinite otherwise. Similarly, positive $C^\infty$ functions with compact support may have an infinite integral against $\delta(f)$. As an illustration, take $n=2$, $f(y)=y_1y_2$, and $\varphi\colon \mathbb{R}^2\to \mathbb{R}$ a $C^\infty$ function, positive with compact support and such that $\varphi(0,0)>0$ (i.e., $\varphi$ does not vanish in a neighborhood of the intersection $y_1=y_2=0$). Then
\[ \int_{\mathbb{R}^2}\delta(f) \varphi =\int_{\mathbb{R}} \varphi(0,y_2) \frac{{\rm d}y_2}{|y_2|} + \int_{\mathbb{R}} \varphi(y_1,0) \frac{{\rm d}y_1}{|y_1|}=+\infty .\]
This is the reason why for a generic $f\in \mathcal{H}(\mathbb{R}^n)$, $n>1$, the measure~$\delta(f)$ cannot be interpreted as a distribution and distribution theory on~$\mathbb{R}^n$ cannot be the natural setting.

It will be clear from the discussion in Remark~\ref{rem:simdif} below that a positive $\sigma$-finite measure $\delta(f)$ can be defined for more general~$f$'s, but
each of the functions $P_x$, $Q_x$, $R$ involved in equation~\eqref{eq:deltaident} is the product of pairwise non-proportional affine functions on some~$\mathbb{R}^n$, so our construction is enough to make sense of $\delta(P_x)$, $\delta(Q_x)$, $\delta(R)$, and we have achieved the preliminary goal of making sense of both sides of the $\delta$-identity equation~\eqref{eq:deltaident}.

We apply the definitions to an explicit computation illustrating their naturalness. In what follows, $a$~is always a fixed parameter. Recall that $P_x(u):=a\prod_{\alpha \in A}(x-u_\alpha)$, a notation suggesting that~$x$ is a (real) parameter and $u\in \mathbb{R}^A$ is the variable. But we could also denote the same expression by $P_u(x)$, a notation suggesting that $u$ is a (vector) parameter and $x\in \mathbb{R}$ is the variable. Finally, we could denote the same expression by $P(x,u)$, a function on $\mathbb{R} \times \mathbb{R}^A$. No matter how we split between parameters and variables, $\prod_{\alpha \in A}(x-u_\alpha)$ is a product of $|A|$ affine functions. We compare $\delta(P_x)$, $\delta(P_u)$ and $\delta(P)$.

-- The function $P_x$ is a product of $|A|$ pairwise non-proportional affine functions of $u\in \mathbb{R}^{A}$. We may use $(u_{\alpha'}-x\delta_{\alpha,\alpha'})_{\alpha'\in A}$ as coordinates near the zero locus of $x-u_\alpha$. We infer from the definitions that, if $\varphi$ is a positive Borel function on~$\mathbb{R}^A$,
\begin{gather} \label{eq:deltaPx} \int \delta(P_x) \varphi = \frac{1}{|a|} \sum_{\alpha\in A} \int_{\mathbb{R}^{A \setminus \{\alpha\}}} \frac{1}{\prod\limits_{\alpha'\neq \alpha} |x-u_{\alpha'}|} \varphi(u)_{|u_\alpha=x} \prod_{\alpha'\neq \alpha } {\rm d}u_{\alpha'}.\end{gather}
This equation (and its counterpart for $Q_x$) will be the starting point of our proof of the $\delta$-identity equation~\eqref{eq:deltaident}.

-- The polynomial (in $x$) $P_u(x)$ is regular if and only if the $u_\alpha$s are pairwise distinct, which is also the condition under which the affine functions $f_\alpha(x):=x-u_\alpha$ of $x$ are pairwise non-proportional and all definitions of $\delta(P_u)$ lead to
\[ \int \delta(P_u) \varphi = \frac{1}{|a|} \sum_{\alpha\in A} \frac{1}{\prod\limits_{\alpha'\neq \alpha} |u_\alpha-u_{\alpha'}|} \varphi(u_\alpha) = \sum_{\alpha\in A} \frac{1}{|P'_u(u_\alpha)|} \varphi(u_\alpha)\]
for any positive Borel function $\varphi$ on $\mathbb{R}$, where $P'_u$ denotes the derivative of $P_u$ with respect to~$x$.

-- The function $P$ is a product of $|A|$ pairwise non-proportional affine functions
on $\mathbb{R} \times \mathbb{R}^A$. We infer from the definitions that if $\varphi$ is a positive Borel function on that product space
\[ \int \delta(P) \varphi =\frac{1}{|a|} \sum_{\alpha\in A} \int_{\mathbb{R}^A} \frac{1}{\prod\limits_{\alpha'\neq \alpha} |u_{\alpha}-u_{\alpha'}|} \varphi(u_\alpha;(u_{\alpha'})_{\alpha'\in A}) \prod_{\alpha'\in A} {\rm d}u_{\alpha'},\]
or equivalently
\[ \int \delta(P) \varphi = \frac{1}{|a|} \sum_{\alpha\in A} \int_{\mathbb{R}\times \mathbb{R}^{A \setminus \{\alpha\}}} \frac{1}{\prod\limits_{\alpha'\neq \alpha} |x-u_{\alpha'}|} \varphi(x;(u_{\alpha'})_{\alpha'\in A})_{|u_\alpha=x} {\rm d}x \prod_{\alpha'\neq \alpha } {\rm d}u_{\alpha'}.\]
The first formula emerges if one uses $(x-u_\alpha;(u_\alpha')_{\alpha'\in A})$ while the second formula emerges if one uses $(x;(u_\alpha'-x\delta_{\alpha,\alpha'})_{\alpha'\in A})$ as coordinates on $\mathbb{R} \times \mathbb{R}^A$ near the zero locus of $x-u_\alpha$.

\begin{remark}
 It is clear that all these formul\ae\ are pretty much one and the same, though the one for $P_u$ involves the restriction that the components of the vector $u$ be pairwise distinct (unless we are prepared to deal with non $\sigma$-finite measures).
\end{remark}

We finally turn to a comparison of the construction given above with the one in~\cite{GS}. The following cursory discussion is used nowhere in the sequel and the uninterested reader can jump directly to Section~\ref{sec:mainproof}. The discussion involves the generalization of the notion of regular function for functions from $\mathbb{R}^n$ to $\mathbb{R}$.

\begin{definition} A function $f\colon \mathbb{R}^n \to \mathbb{R}$ is called regular if $f$ is $C^1$ and its gradient $\operatorname{Grad} f$ vanishes nowhere on $f=0$.
\end{definition}

\begin{remark} \label{rem:simdif} In \cite{GS}, Gelfand and Shilov define $\delta(f)$ (which we denote by $\delta^{\rm GS}$ to make later comparison with our definitions easier and avoid confusion) as a generalized function when $f\colon \mathbb{R}^n \to \mathbb{R}$ is an arbitrary regular function. We have already stressed that this is not an auto\-matic consequence of the definition of $\delta_0$ (either as a measure or as a distribution). However (maybe modulo the inclusion of absolute values of Jacobians, see the first point below) the definition chosen by Gelfand and Shilov is the ``only'' natural one, the one that can be manipulated most intuitively. For instance, it behaves nicely if some variables in $f$ are treated as parameters. We have seen an instance of this in our context when we compared $\delta(P_x)$, $\delta(P_u)$ and~$\delta(P)$.

 Another view of the naturalness properties of the definition of $\delta^{\rm GS}(f)$ is that if $\theta_k\colon \mathbb{R}\to \mathbb{R}$, $k\in \mathbb{N}$ is a sequence of smooth functions converging towards $\delta_0$ in the distribution topology (on~$\mathbb{R}$) then the compositions $\theta_k(f)=\theta_k\circ f$ converge towards $\delta(f)$ in the distribution topology (on~$\mathbb{R}^n$). We have seen a trivial counterpart of this in Remark~\ref{rem:conv}. The ``approximation of distribution by smooth maps'' procedure is in fact the path followed by H\"ormander in \cite[Chapter~6 and Section~8.2]{H} to define the composition of a general distribution on $\mathbb{R}^m$ with functions from $\mathbb{R}^n$ to $\mathbb{R}^m$ satisfying appropriate conditions (surjectivity of the differential everywhere in Chapter~6, less stringent conditions allowing certain singularities in Section~8.2).

The reader is invited to consult~\cite{GS} for the detailed definition of $\delta^{\rm GS}$ and other distributions localized of $f=0$. Coming back to the comparison, one can check the following:

-- First, in one dimension, $\delta^{\rm GS}(f)$ is the distribution associated to the measure~$\delta(f)$ from Definition~\ref{def:dimone}, with one little proviso: Gelfand and Shilov do not include an absolute value for the derivatives in the denominators. More generally, they do not include absolute values for Jacobians because they use the framework of differential forms, whereas densities better fit our needs. When we talk of~$\delta^{\rm GS}$ in the sequel, we always have in mind that absolute values are included.

-- Second, an affine function $f\colon \mathbb{R}^n \to \mathbb{R}$ is regular if and only if it is non-constant, and then $\delta^{\rm GS}(f)$ is the distribution associated to the measure $\delta(f)$ from Definition~\ref{def:aff}. In fact, the definition of $\delta^{\rm GS}(f)$ for a general regular function follows the same pattern, using localization and replacing the linear change of coordinates by the implicit function theorem. Borrowing the notations from Definition~\ref{def:aff}, $z_1=f(y),z_2,\dots,z_n$ become arbitrary local coordinates, the Jacobian is not a constant anymore and the formula for $\langle \delta^{\rm GS}(f), \varphi \rangle$ is similar to equation~\eqref{eq:aff}, but with $\operatorname{Jac}(y,z)$ replaced by $\operatorname{Jac}(y,z)_{z_1=0}$. Using a partition of unity, one may assume that the local coordinates are well-defined in an open set containing the support of $\varphi$ so that $\tilde{\varphi}=\varphi \circ z^{-1}$ extends as a $C^\infty$ function with compact support. The consistency of the procedure is guaranteed by the general change of variable formula. We refer again to~\cite{GS} for all details.

-- Third, the identity equation~\eqref{eq:deltaprod} (with $\delta^{\rm GS}$ substituted for $\delta$, and interpreted as an identity between distributions) remains true whenever $g,f,fg\colon \mathbb{R}^n \to \mathbb{R}$ are regular. Then the analog of equation~\eqref{eq:deltamultiprod} also holds. Thus, if $f \in \mathcal{H}(\mathbb{R}^n)$ is regular, $\delta^{\rm GS}(f)$ is the distribution associated to the measure $\delta(f)$ from Definition \ref{def:delta-prod}. It is immediate that an element of $\mathcal{H}(\mathbb{R}^n)$ is regular if and only if all factors in a factorization have the same linear part but pairwise distinct constant terms. As already noticed, all members of $\mathcal{H}(\mathbb{R})$ are regular, but generic members of $\mathcal{H}(\mathbb{R}^n)$, $n>1$, are not. For such members $f$, the measure $\delta(f)$ assigns an infinite measure to certain bounded Borel subsets of $\mathbb{R}^n$, and $\delta(f)$ cannot be interpreted as a distribution.

-- Fourth, however, every $f\in \mathcal{H}(\mathbb{R}^n)$ with factorization $(f_1,\dots,f_k)$ is regular on a dense connected open subset of $\mathbb{R}^n$, namely the complement $O_f$ of $\cup_{i\neq j} \{f_i=f_j=0\}$. The Gelfand--Shilov construction defines a distribution (which we still denote by $\delta^{\rm GS}(f)$) on $O_f$, which is positive on positive test functions (with support in $O_f$ by definition), hence defines a positive $\sigma$-finite measure on $O_f$, and on $O_f$ equation~\eqref{eq:deltamultiprod} again holds. Hence this measure coincides with $\delta(f)$ (as defined in Definition \ref{def:delta-prod}) on Borel subsets of $O_f$. Consequently $\delta(f)$ can also be interpreted as an extension of (the positive $\sigma$-finite measure on $O_f$ associated to) $\delta^{\rm GS}(f)$ to the whole of $\mathbb{R}^n$. The ambiguity is a positive $\sigma$-finite measure with support in $\cup_{i\neq j} \{f_i=f_j=0\}$, and Definition \ref{def:delta-prod} consists in extending by the zero measure. We shall not try to explore whether some other extensions would preserve the nice properties of the definition of $\delta(f)$, in particular whether some generalized version of equation~\eqref{eq:deltaident} could be obtained in that way. Finally we note that extension by the zero measure could be used much more generally to extend a (positive $\sigma$-finite) measure defined on some Borel subset of a measurable space to the whole space, as is customary done for the similar construction of extensions for distributions.
\end{remark}

\section{Proof of Proposition \ref{prop:main}} \label{sec:mainproof}

\begin{proof} To avoid clumsy notation, we write simply $P(x)$ (resp.\ $Q(x)$) for $P_u(x)$ (resp.\ $Q_v(x)$) in this proof.

We start with the proof of $(i)$. Recall that
\[ J(u,v)=a^{|B|-1} b^{|A|-1} \sum_{\beta' \in B} \prod\limits_{\beta'' \neq \beta'}
 \frac{\prod\limits_{\alpha''\in A}(v_{\beta''}-u_{\alpha''})}{(v_{\beta'}-v_{\beta''})}.\]

Rewrite
\begin{gather*} a^{|B|-1} b^{|A|-1}\!\prod_{\beta'' \neq \beta'}\frac{\prod\limits_{\alpha''\in A}(v_{\beta''}-u_{\alpha''})}{(v_{\beta'}-v_{\beta''})} = b^{|A|-1} \prod_{\beta''\ne \beta'}\frac{P(v_{\beta''})}{(v_{\beta'}-v_{\beta''})} = \frac{b^{|A|}}{P(v_{\beta'})Q'(v_{\beta'})} \prod_{\beta''\in B}P(v_{\beta''})\end{gather*}
and notice that if $S(x):=P(x)Q(x)$, then $P(v_{\beta'})Q'(v_{\beta'})=S'(v_{\beta'})$ to conclude that
\begin{gather} J(u,v) = b^{|A|} \left(\sum_{\beta'\in B} \frac{1}{S'(v_{\beta'})}\right) \prod_{\beta''\in B} P(v_{\beta''})
= b^{|A|}a^{|B|}\left(\sum_{\beta'\in B} \frac{1}{S'(v_{\beta'})}\right) \prod_{\alpha'' \in A \atop \beta''\in B} (v_{\beta''}-u_{\alpha''}).\!\!\! \label{eq:Jprod}\end{gather}

The same manipulation on
\[ \tilde{J}(u,v)= a^{|B|-1} b^{|A|-1} \sum_{\alpha' \in A} \prod_{\alpha'' \neq \alpha'} \frac{\prod\limits_{\beta''\in B}(u_{\alpha''}-v_{\beta''})}{(u_{\alpha'}-u_{\alpha''})}\]
leads to
\begin{gather} \tilde{J}(u,v) = a^{|B|} \left(\sum_{\alpha''\in A} \frac{1}{S'(u_{\alpha'})}\right)\! \prod_{\alpha''\in A} \! Q(v_{\alpha''}) = a^{|B|}b^{|A|}\left(\sum_{\alpha''\in A} \frac{1}{S'(u_{\alpha'})}\right) \! \prod_{\beta'' \in B\atop \alpha''\in A}\! (u_{\alpha''}-v_{\beta''}).\!\!\! \label{eq:tJprod}\end{gather}
But if $S=c\prod_{\gamma \in C} (x-w_\gamma)$ is an arbitrary polynomial of degree at least $2$ with simple zeros, one has
the ``well known'' identity\footnote{For instance, note that $\frac{1}{S(x)}=\sum_{\gamma \in C} \frac{1}{S'(w_\gamma)}\frac{1}{x-w_\gamma}$ on the one hand, and $\frac{1}{S(x)}=0\big(1/x^2\big)$ at infinity.} $\sum_{\gamma \in C} \frac{1}{S'(w_\gamma)}=0$. For the polynomial $S$ at hand this entails
\[ \sum_{\alpha' \in A} \frac{1}{S'(u_{\alpha'})}=-\sum_{\beta' \in B}\frac{1}{S'(v_{\beta'})}.\]
Comparison of equations~\eqref{eq:Jprod} and~\eqref{eq:tJprod} leads immediately to $J=(-)^{|A||B|-1}\tilde{J}$ concluding the proof of~$(i)$.

To prove $(ii)$ we simply note that the original formula for $J$ is a polynomial in~$u$ with rational coefficients in $v$ while the
original formula for $\tilde{J}$ is a polynomial in $v$ with rational coefficients in~$u$. As~$J$ and $\tilde{J}$ differ (at most) by a sign, altogether $J$ is a polynomial in $(u,v)$.

We turn finally to the proof of $(iii)$. We rewrite equation~\eqref{eq:deltaPx} in a number of more compact forms, each of which we shall use freely, namely
\begin{gather*} \delta(P_x) = \frac{1}{|a|} \sum_{\alpha\in A} \frac{1}{\prod\limits_{\alpha'\neq \alpha} |x-u_{\alpha'}|} \delta_x(u_\alpha) \prod_{\alpha'\neq \alpha } {\rm d}u_{\alpha'}= \sum_{\alpha\in A} \frac{1}{|P'(u_\alpha)|} \delta_x(u_\alpha) \prod_{\alpha'\neq \alpha } {\rm d}u_{\alpha'}\\
\hphantom{\delta(P_x)}{} = \frac{1}{|a|} \sum_{\alpha\in A} \frac{1}{\prod\limits_{\alpha'\neq \alpha} |x-u_{\alpha'}|} \delta(u_\alpha-x)= \sum_{\alpha\in A} \frac{1}{|P'(u_\alpha)|}\delta(u_\alpha-x).
\end{gather*}
In the first line, $\delta_x$ is the one-dimensional $\delta$ function for a unit mass at point $x$ along the coordinate axis $\alpha$, while in the second line $\delta(u_\alpha-x)$ is a distribution on $\mathbb{R}^A$, namely $\delta(f)$ for the affine function $f\colon \mathbb{R}^A \to \mathbb{R}$, $u \mapsto f(u):=u_\alpha-x$. Of course, the definition of~$\delta(f)$ yields $\delta(u_\alpha-x)=\delta_x(u_\alpha)\prod_{\alpha'\neq \alpha } {\rm d}u_{\alpha'}$. To get the second equality in each line, we have used that $a\prod_{\alpha'\neq \alpha} (x-u_{\alpha'})$ and $a\prod_{\alpha'\neq \alpha} (u_\alpha-u_{\alpha'})=P'_u(u_\alpha)$ are equal on the support of $\delta_x(u_\alpha)$ or of $\delta(u_\alpha-x)$.

Now $\delta(P_x)\delta(Q_x)$, which could be rewritten more carefully as $\delta(P_x)\otimes \delta(Q_x)$, is well defined as the product of two measures, $\delta(P_x)$ on $\mathbb{R}^A$ and $\delta(Q_x)$ on $\mathbb{R}^B$. When $P$ and $Q$ are monic of degree $1$, $P=(x-u)$, $Q=(x-v)$,
\[ \int {\rm d}x \int \varphi(u,v) \delta_x(u) \delta_x(v)= \int {\rm d}x \, \varphi(x,x)\]
holds for every positive Borel function $\varphi$ on~$\mathbb{R}^2$, and the very definition of $\delta(u-v)$ as a positive measure on $\mathbb{R}^2$ ensures that $\int {\rm d}x\, \delta_x(u) \delta_x(v)=\delta(u-v)$. Thus~$(iii)$ holds in this special case, and we shall use it to prove the general case. Namely we translate the special case
\[ \int_x {\rm d}x\, \delta_x(u_\alpha)\delta_x(v_\beta) = \delta(u_\alpha- v_\beta),\]
where $\delta(u_\alpha- v_\beta)$ is to be interpreted as the measure in the plane $(u_\alpha,v_\beta)\in\mathbb{R}^2$, into
\[ \int {\rm d}x\, \delta(u_\alpha-x)\delta(v_\beta-x) = \delta(u_\alpha- v_\beta),\]
where now $\delta(u_\alpha-x)$ is interpreted as the measure $\delta(f)$ on $\mathbb{R}^A$ for $f(u)=u_\alpha-x$, $\delta(v_\beta-x)$ is interpreted as the measure $\delta(g)$ on $\mathbb{R}^B$ for $g(v)=v_\beta-x$, and $\delta(u_\alpha- v_\beta)$ as the measure $\delta(h)$ on $\mathbb{R}^{A\cup B}$ for $h(u,v)=u_\alpha- v_\beta$. Thus
\begin{align}
\int {\rm d}x], \delta(P_x)\delta(Q_x) & = \sum_{\alpha\in A \atop \beta \in B} \int {\rm d}x\, \frac{1}{|P'(u_\alpha)|} \delta(u_\alpha-x) \frac{1}{|Q'(v_\beta)|} \delta(v_\beta-x) \nonumber \\
& = \sum_{\alpha\in A \atop \beta \in B} \frac{1}{|P'(u_\alpha)||Q'(v_\beta)|} \delta(u_\alpha- v_\beta)\label{eq:delPQdx} \\
& = \sum_{\alpha\in A \atop \beta \in B} \frac{1}{|P'(v_\beta)||Q'(u_\alpha)|} \delta(u_\alpha- v_\beta),\label{eq:delPQdxs}
\end{align}
where in the last line we have used that $\frac{1}{P'(u_\alpha)Q'(v_\beta)}$ and $\frac{1}{P'(v_\beta)Q'(u_\alpha)}$ coincide on the support~$\Sigma_{\alpha,\beta}$ of the measure $\delta(u_\alpha- v_\beta)$.
On the other hand, using Definition~\ref{def:delta-prod} for $R(u,v)$ we get
\begin{gather} \label{eq:delR}\delta(R)=\sum_{\alpha\in A \atop \beta \in B} \frac{1}{|J_{\alpha,\beta}(u,v)|} \delta(u_\alpha- v_\beta),\end{gather}
where
\[ J_{\alpha,\beta}(u,v):=a^{|B|} b^{|A|}\prod_{(\alpha',\beta')\neq (\alpha,\beta)} (u_{\alpha'}- v_{\beta'}). \]
Observe that, on $\Sigma_{\alpha,\beta}$,
\[ a\prod_{\alpha' \neq \alpha} (u_{\alpha'}- v_{\beta})= (-)^{|A|-1}P'(v_\beta), \qquad b\prod_{\beta' \neq \beta} (u_{\alpha}- v_{\beta'})= Q'(u_\alpha), \]
so that, again on $\Sigma_{\alpha,\beta}$,
\begin{gather} \label{eq:jab} J_{\alpha,\beta}(u,v) = a^{|B|-1} b^{|A|-1} \prod_{\alpha'' \neq \alpha \atop \beta''\neq\beta} (u_{\alpha''}- v_{\beta''}) (-)^{|A|-1}P'(v_\beta) Q'(u_\alpha).\end{gather}
Fix $\alpha \in A$ and $\beta\in B$. If $\beta'\neq \beta$ then $\prod_{\beta'' \neq \beta'} \prod_{\alpha''\in A}(v_{\beta''}-u_{\alpha''})$ contains a factor $v_{\beta}-u_{\alpha}$. Thus, if $v_{\beta}=u_{\alpha}$, all terms but the one corresponding to $\beta'=\beta$ in the sum defining $J(u,v)$ vanish. Hence, on~$\Sigma_{\alpha,\beta}$, we have
\begin{align*}
J(u,v) & = a^{|B|-1} b^{|A|-1}\prod_{\beta'' \neq \beta}
 \frac{\prod\limits_{\alpha''\in A}(v_{\beta''}-u_{\alpha''})}{(v_{\beta}-v_{\beta''})}
 = a^{|B|-1} b^{|A|-1}\prod_{\beta'' \neq \beta}\frac{\prod\limits_{\alpha''\in A}(v_{\beta''}-u_{\alpha''})}{(u_\alpha-v_{\beta''})}\\
 & = a^{|B|-1} b^{|A|-1} (-)^{|B|-1}\prod_{\alpha''\neq \alpha \atop \beta'' \neq \beta} (v_{\beta''}-u_{\alpha''}).
\end{align*}
Using equation~\eqref{eq:jab}, we obtain that, on $\Sigma_{\alpha,\beta}$,
\[\frac{|J(u,v)|}{|J_{\alpha,\beta}(u,v)|}=\frac{1}{|P'(v_\beta)||Q'(u_\alpha)|}.\]
Comparison of equations~\eqref{eq:delPQdxs} and~\eqref{eq:delR} establishes the validity of~$(iii)$.
\end{proof}

\begin{remark} The explicit expression of the polynomial $J(u,v)$ is quite complicated in general. It is rather simple in the case when $Q$ is of degree~$2$. Then, writing $v=(v',v'')$ i.e., $Q_v(x)=b(x-v')(x-v'')$ one obtains
\[ J(u,v)=-b^{|A|-1}\frac{P(v')-P(v'')}{v'-v''},\]
which is a simple divided difference, hence is obviously polynomial in~$(u,v)$.
Even the case when~$P$ and~$Q$ are of degree $3$ leads to a prohibitively complicated explicit expression for $J(u,v)$ as a~polynomial.
\end{remark}

\begin{remark} It is a simple consequence of equation~\eqref{eq:delPQdx} that integration over $x$ does not introduce unexpected singularities: the measure of a compact set $K$ in $\mathbb{R}^{A\cup B}$ under $\int {\rm d}x\, \delta(P_x)\delta(Q_x)$ is finite if~$K$ does not meet any hyperplane $u_{\alpha}=u_{\alpha'}$, $\alpha\neq \alpha'$ or $v_{\beta}=v_{\beta'}$, $\beta\neq \beta'$.
\end{remark}

\subsection*{Acknowledgements}

We thank Michel Talagrand for his comments and encouragement.

\pdfbookmark[1]{References}{ref}
\LastPageEnding

\end{document}